\documentclass{aims}
\usepackage{amsmath}
  \usepackage{paralist}
  \usepackage{graphics} %% add this and next lines if pictures should be in esp format
  \usepackage{epsfig} %For pictures: screened artwork should be set up with an 85 or 100 line screen
 \usepackage[colorlinks=true]{hyperref}
\hypersetup{urlcolor=blue, citecolor=red}

\usepackage{mathtools}

\newtheorem{theorem}{Theorem}%[section]

\theoremstyle{definition}

\newcommand{\R}{\mathbb{R}}

\newcommand{\bu}{\mathbf{u}}
\newcommand{\bA}{\mathbf{A}}
\newcommand{\bb}{\mathbf{b}}
\newcommand{\0}{\mathbf{0}}
\newcommand{\bS}{\mathbf{S}}
\newcommand{\bv}{\mathbf{v}}
\newcommand{\ba}{\mathbf{a}}
\newcommand{\bc}{\mathbf{c}}
\newcommand{\mopt}{\ensuremath{M_\text{opt}}}
\newcommand{\bx}{\mathbf{x}}
\newcommand{\by}{\mathbf{y}}
\def\bs{\mathbf{s}}

\title[Why Optimal States Recruit Fewer Reactions]
      {Why Optimal States Recruit Fewer Reactions in Metabolic Networks}

\author[Joo Sang Lee, Takashi Nishikawa, Adilson E. Motter]{}


 \email{joosang@northwestern.edu}
 \email{tnishika@clarkson.edu}
 \email{motter@northwestern.edu}

\begin{document}
\maketitle

% Enter the first author's name and address:
\centerline{\scshape Joo Sang Lee }
\medskip
{\footnotesize
% please put the address of the first author
 \centerline{ Department of Physics \& Astronomy, Northwestern University, Evanston, IL 60208, USA}
} % Do not forget to end the {\footnotesize by the sign }

\medskip

\centerline{\scshape Takashi Nishikawa }
\medskip
{\footnotesize
 % please put the address of the second  and third author
 \centerline{ Department of Mathematics, Clarkson University,  Potsdam, NY 13699, USA}
}

\medskip

\centerline{\scshape Adilson E. Motter }
\medskip
{\footnotesize
 \centerline{ Department of Physics \& Astronomy and Northwestern Institute on Complex Systems,}
 \centerline{Northwestern University, Evanston, IL 60208, USA}
 \centerline{ Department of Molecular Biology, Princeton University, Princeton, NJ 08544, USA}
} 

\bigskip

\begin{abstract}
The metabolic network of a living cell involves several hundreds or thousands of interconnected biochemical reactions.  Previous research has shown that under realistic conditions only a fraction of these reactions is concurrently active in any given cell. This is partially determined by nutrient availability, but is also strongly dependent on the metabolic function and network structure.  Here, we establish rigorous bounds showing that the fraction of active reactions is smaller (rather than larger) in metabolic networks evolved or engineered to optimize a specific metabolic task, and we show that this is largely determined by the presence of thermodynamically irreversible reactions in the network.
We also show that the inactivation of a certain number of reactions determined by irreversibility can generate a cascade of secondary reaction inactivations that propagates through the network.
The mathematical results are complemented with numerical simulations of the
metabolic networks of  the bacterium \emph{Escherichia coli} and of human cells,
which show, counterintuitively,  that even the maximization of the total reaction flux in the network leads to a reduced number of active reactions.
\end{abstract}

\section{Introduction}
\label{sec1}

The mathematical modeling of biological networks has focused on the influence of the network structure on the functional properties of the system \cite{alon2006,bar2004,pals2006}.
Insights provided by these studies have shown, for example, that structural modules and hierarchical organization in the network are often related to compartmentalization of functional processes \cite{ravasz2002,spirin2003}.  This is important since intracellular processes are rarely carried out by individual elements, and often involve the coordinated activity of multiple genes, proteins, and biochemical transformations. Because different components may be recruited for different processes, the most fundamental aspect of the dynamics of complex intracellular networks concerns precisely the characterization of the specific parts of the network that are active under given conditions.

Recent research focused on the modeling of metabolic networks has shown that typical metabolic states tend to recruit a much larger number of reactions than states that maximize growth rate \cite{Takashi2008}. This counterintuitive property is important in multiple contexts. For example, it provides a partial explanation for the apparent dispensability of a large fraction of genes in single-cell organisms \cite{papp2004,blank2005,Cornelius2011},  since the genes associated with reactions that become inactive in growth-maximizing states are expected not to be essential. This also explains why genetic and environmental perturbations that cause growth defect are accompanied by a burst in reaction activity \cite{Fong2005,Fong2006}. These bursts can be attributed to the transient activation of otherwise inactive reactions that are recruited by the suboptimal states that follow the perturbation \cite{Takashi2008}, since such states tend to have a larger number of active reactions. Another important implication concerns the possibility of synthetic rescues \cite{Motter2008, Kim2009}, where the inactivation of one gene can be compensated by the targeted inactivation of other genes. The inactivation of such rescue genes can thus allow the recovery of lost biological function. This is possible in part because the rescue genes correspond to genes that would be inactive in an optimal state,  so that disabling them helps bring the state of the system closer to the desired optimal state \cite{Motter2010}. Understanding the root causes of the reduced reaction activity in optimal metabolic states is then of significant interest in the characterization and study of cellular metabolism.

Focusing on steady-state dynamics, here we use flux-balance based analysis and linear programming techniques to establish rigorous results on the number of reactions that can be active in a given metabolic network.  We derive conditions for a specific reaction to be inactive in all (Sec.\ \ref{sec3})  or active in almost all (Sec.\ \ref{sec4}) feasible metabolic states. We also establish bounds for the number of reactions that can be active in states that optimize an arbitrary linear function of reaction fluxes (Sec.\ \ref{sec5}),
which are derived based on the duality principle in linear programming,
and we study the uniqueness of the optimal solution for typical linear objective functions (Sec.\ \ref{sec6}). Finally, we implement numerical simulations in
reconstructed \emph{Escherichia coli} and human metabolic networks,
 both to compare with the rigorous bounds and to consider nonlinear objective functions (Sec.\ \ref{sec7}). Taken together, our results show that the reduced number of active reactions in the optimal states of a linear objective function, including growth rate, are mainly determined by the presence of irreversible reactions in the network. The irreversibility constraints are shown to play a role also in nonlinear objective functions, such as the aggregated flux and mass flow activity, whose optimal solutions are shown to have a number of active reactions comparable to the corresponding number for linear functions.

\section{Preliminary remarks}
\label{sec2}

\noindent
We consider time-independent metabolic states, which serve as an appropriate representation of the
state of single cells at time scales much shorter than the lifetime of the cells,  as well as of the average
behavior of a large population of cells  at arbitrary time scales  in time-invariant conditions.
Under this  \emph{steady-state} assumption, a cellular metabolic state is a solution of a homogeneous linear equation that accounts for
all stoichiometric constraints,
\begin{equation}\label{eqn:mb}
 \mathbf{S} \mathbf{v}=\mathbf{0},
\end{equation}
where $\mathbf{S}$ is the $m \times N$ stoichiometric matrix and $\bv \in \R^N$ is the vector of metabolic fluxes.
The components of $\bv = (v_1,\dots,v_N)^T$ include the fluxes of $n$ internal and transport reactions as well as
$n_\text{ex}$ exchange fluxes,
which model the transport of metabolic species across the system boundary.
Constraints of the form $v_i \le \beta_i$ imposed on the exchange fluxes are used to limit the
maximum uptake rates of substrates in the medium.
Additional constraints of the form $v_i \ge 0$ arise for the reactions that are irreversible.
Assuming that the cell's operation is mainly limited by the availability of substrates in the medium,
we impose no other constraints on the internal reaction fluxes, except for the ATP maintenance flux,
which is set to a fixed positive value. These additional constraints can be organized in the form
\begin{equation}\label{eqn:const}
\alpha_i \le v_i \le \beta_i,\quad i=1,\ldots,N.
\end{equation}
The set of all flux vectors $\bv$ satisfying Eqs.\ (\ref{eqn:mb}) and (\ref{eqn:const}) defines
the \emph{feasible solution space} $M \subset \R^N$, representing the capability of the metabolic network as a
system.  Because the number of fluxes $N$ is larger than the number of metabolic species $m$,
system (\ref{eqn:mb}) is under-determined and $M$ is  generally high dimensional.

Our study is formulated in the context of flux balance analysis~\cite{bonarius1997,varma1994},
which is based on the maximization of a metabolic objective function $\bc^T \bv$ within the feasible solution space $M$ (the superscript $T$ is used to denote transpose).
This reduces to
a linear programming problem
\begin{equation}\label{lp}
\begin{alignedat}{2}
&\text{maximize: } & \quad & \bc^T \bv = \sum_{i=1}^N c_i v_i\\
&\text{subject to: } & &\bS \bv = \0, \quad \bv  \in \R^N,\\
& & &\alpha_i \le v_i \le \beta_i,\quad i=1,\ldots,N,
\end{alignedat}
\end{equation}
where we set $\alpha_i = -\infty$ if $v_i$ is not bounded from below and  $\beta_i = \infty$ if $v_i$ is not bounded from above.
For a given objective function,
we can numerically determine an optimal flux
distribution for this problem.
This formulation is
also appropriate for the derivation of the rigorous results presented below.
In the particular case of growth maximization, the objective
vector $\bc$ is taken to be parallel to the biomass flux, which is modeled as an
effective reaction that converts metabolic species into biomass.

%--------------------------------------------------------------------------------------------------
\section{Inactivity due to stoichiometric constraints}
\label{sec3}

\noindent
The first question of interest is to determine the conditions under which a reaction will be inactive for any solution of Eq.\ (\ref{eqn:mb}).
Let us define the stoichiometric coefficient vector of reaction $i$ to be the $i$th column of the stoichiometric matrix $\bS$.
We similarly define the stoichiometric coefficient vector of an exchange flux.
If the stoichiometric vector of reaction $i$ can be written as a linear combination of the stoichiometric vector of reactions/exchange fluxes $i_1, i_2, \ldots, i_k$, we say that $i$ is a linear combination of $i_1, i_2, \ldots, i_k$.
We use this linear relationship to completely characterize the set of all reactions that are always inactive due to the stoichiometric constraints, regardless of any additionally imposed constraints, such as the availability of substrates in the medium, reaction irreversibility, cell maintenance requirements, and optimum growth condition.

\begin{theorem}\label{thm1}
Reaction $i$ is inactive for all $\bv$ satisfying $\bS \bv = \0$ if and only if it is not a linear combination of the other reactions and exchange fluxes.
\end{theorem}

\begin{proof}
We denote the stoichiometric coefficient vectors of reactions and exchange fluxes by $\bs_1, \ldots, \bs_N$.
The theorem is equivalent to saying that there exists $\bv$ satisfying both $\bS \bv = \0$ and $v_i \neq 0$ if and only if $\bs_i$ is a linear combination of $\bs_k$, $k=1,2,\ldots,N$, $k \neq i$.

To prove the forward direction in this statement, suppose that $v_i \neq 0$ in a state $\bv$ satisfying $\bS \bv = \0$.
By writing out the components of the equation $\bS \bv = \0$ and rearranging, we get
\begin{equation}
s_{ji} v_i = \sum_{k \neq i} (- v_k) s_{jk}, \quad j=1,\ldots,m. \nonumber
\end{equation}
Since $v_i \neq 0$, we can divide this equation by $v_i$ to see that $\bs_i$ is a linear combination of $\bs_k$, $k \neq i$ with coefficients $c_k = - v_k/v_i$.

To prove the backward direction, suppose that $\bs_i = \sum_{k \neq i} c_k \bs_k$.
If we choose $\bv$ so that $v_k = c_k$ for $k\neq i$ and $v_i = -1$, then for each $j$, we have
\[ (\bS \bv)_j = \sum_k v_k s_{jk}
= -s_{ji} + \sum_{k\neq i} c_k s_{jk}
= 0, \]
so $\bv$ satisfies $\bS \bv = \0$.
\end{proof}

Theorem \ref{thm1} holds true independently of other constraints because the proof does not involve Eq.\ (\ref{eqn:const}).
In particular, the sufficient condition for inactivity applies to any nutrient medium condition and does not depend on the reversibility of the reactions under consideration.
In the case of the reconstructed \emph{E.\ coli}
(human) metabolic network considered in this study (described in Sec.\ \ref{sec7}),
which includes 922 (3328) unique internal and transport reactions,  a total of 141 (475) reactions are always inactive in
steady states as a result of the condition in Theorem \ref{thm1}.

%--------------------------------------------------------------------------------------------------
\section{Activity in typical steady states}
\label{sec4}

\noindent
The next question of interest concerns the number of reactions that will be active with probability one in typical metabolic states.
The stoichiometric constraints $\bS \bv = \0$ define the linear subspace $\text{Nul}\,\bS = \{ \bv \in \R^N \,|\, \bS \bv = \0 \}$ (the null space of $\bS$), which contains the feasible solution space $M$.
However, the set $M$ can possibly be smaller than $\text{Nul}\,\bS$ because of the additional constraints arising from  environmental and physiochemical properties (availability of substrates in the medium, reaction irreversibility, and cell maintenance requirements).
Therefore, $M$ may have smaller dimension than $\text{Nul}\,\bS$.
If we denote the dimension of $M$ by $d$, there exists a unique $d$-dimensional linear submanifold of $\R^N$ that contains $M$, which we denote by $L_M$.
We can then use the Lebesgue measure naturally defined on $L_M$~\cite{rudin1987real} to make probabilistic statements, since we can define the probability of a subset $A \subseteq M$ as the Lebesgue measure of $A$ normalized by the Lebesgue measure of $M$.
In particular, we say that $v_i \neq 0$ for {\em almost all} $\bv \in M$ if the set $\{ \bv \in M \,|\, v_i = 0 \}$ has Lebesgue measure zero on $L_M$.
An interpretation of this is that $v_i \neq 0$ with probability one for an organism in a random state under given environmental conditions, which can be used to  prove the following theorem.
\begin{theorem}\label{thm2}
If $v_i \neq 0$ for some $\bv \in M$, then $v_i \neq 0$ for almost all $\bv \in M$.
\end{theorem}
\begin{proof}
Suppose that $v_i \neq 0$ for some $\bv \in M$. The set $L_i := \{ \bv \in L_M \,|\, v_i = 0 \}$ is a linear submanifold of $L_M$, so we have $\dim{L_i} \le \dim{L_M}$.
If $\dim{L_i} = \dim{L_M}$, then we have $L_i = L_M \supseteq M$, implying that we have $v_i = 0$ for any $\bv \in M$, which violates the assumption.
Thus, we must have $\dim{L_i} < \dim{L_M}$, implying that $L_i$ has zero Lebesgue measure on $L_M$.
Since $M \subseteq L_M$, we have $M_i := \{ \bv \in M \,|\, v_i = 0 \} \subseteq \{ \bv \in L_M \,|\, v_i = 0 \} = L_i$, and thus $M_i$ also has Lebesgue measure zero.
Therefore, we have $v_i \neq 0$ for almost all $\bv \in M$.
\end{proof}

Theorem~\ref{thm2} implies that we can group the reactions and exchange fluxes into two categories:
\begin{enumerate}
\item {\em Always inactive}: $v_i = 0$ for all $\bv \in M$, and
\item {\em Almost always active}: $v_i \neq 0$ for almost all $\bv \in M$.
\end{enumerate}
Consequently, the number $n_+(\bv)$ of active reactions satisfies
\begin{equation}
n_+(\bv) = n_+^\text{typ} := n - n_0^m - n_0^e \quad \text{for almost all } \bv \in M,
\label{eq:env}
\end{equation}
where $n_0^m$ is the number of inactive reactions due to the stoichiometric constraints (characterized by Theorem~\ref{thm1}) and $n_0^e$ is the number of additional reactions in the category 1 above, which are due to the environmental and  irreversibility conditions.
Combining this result with the finding that optimal states have fewer active reactions (next section), it follows that a typical state $\bv \in M$ is non-optimal.

Equation (\ref{eq:env}) will lead to a different number of active reactions for different nutrient medium conditions [determined by Eq.\ (\ref{eqn:const})], with the general trend that this number will be larger in richer medium conditions. In the case of  the  \emph{E.\ coli}
(human) reconstructed network simulated in glucose minimal medium, as considered here (see Sec.\ \ref{sec7}), the number $n_0^e$  of inactive internal and transport reactions is
182 (1274), of which 158 (563)  are due to environmental limitations and 24 (711) are due to reaction irreversibility.
The latter includes the cascading-induced inactivation of some reactions due to the inactivation of different, irreversible reactions.

%--------------------------------------------------------------------------------------------------
\section{Activity in optimal states}
\label{sec5}

\noindent
We now turn to the  central part of our study, which concerns the number of reactions that can be active in steady states that optimize
a linear function of the metabolic fluxes. The linear programming problem for finding the flux distribution maximizing a linear objective function
 can be written in the matrix form:
\begin{equation}\label{primal}
\begin{alignedat}{2}
&\text{maximize: } & \quad & \bc^T \bv \\
&\text{subject to: } & &\bS \bv = \0,\;
\bA \bv \le \bb, \; \bv \in \R^N,
\end{alignedat}
\end{equation}
where $\bA$ and $\bb$ are defined as follows.
If the $i$th constraint is $v_j \le \beta_j$, the $i$th row of $\bA$ consists of all zeros except for the $j$th entry that is $1$, and $b_i = \beta_j$.
If the $i$th constraint is $\alpha_j \le v_j$, the $i$th row of $\bA$ consists of all zeros except for the $j$th entry that is $-1$, and $b_i = - \alpha_j$.
A constraint of the type $\alpha_j \le v_j \le \beta_j$ is broken into two separate constraints and represented in $\bA$ and $\bb$ as above.
The inequality between vectors is interpreted as inequalities between the corresponding components, so if the rows of $\bA$ are denoted by $\ba_1^T,\ba_2^T,\ldots,\ba_K^T$, the inequality $\bA \bv \le \bb$ represents the set of $K$ constraints $\ba_i^T \bv \le b_i$, $i=1,\ldots,K$.  By defining the feasible solution space
\begin{equation}
M := \{ \bv \in \R^N \,|\, \bS \bv = \0,\; \bA \bv \le \bb \},
\end{equation}
the problem can be compactly expressed as maximizing $\bc^T \bv$ in $M$.

The duality principle \cite{math_ref}
expresses that any linear programming problem (primal problem) is associated with a complementary linear programming problem (dual problem), and the solutions of the two problems are intimately related.
The dual problem associated with problem~\eqref{primal} is
\begin{equation}\label{dual}
\begin{alignedat}{2}
&\text{minimize: } & \quad & \bb^T \bu_1 \\
&\text{subject to: } & & \bA^T \bu_1 + \bS^T \bu_2 = \bc, \;
\bu_1 \ge \0, \\
& & &\bu_1 \in \R^K, \; \bu_2 \in \R^m,
\end{alignedat}
\end{equation}
where $\{ \bu_1, \bu_2 \}$ is the dual variable.
A consequence of the Strong Duality Theorem \cite{math_ref} is that the primal and dual solutions are related via a well-known optimality condition: $\bv$ is optimal for problem~\eqref{primal} if and only if there exists $\{ \bu_1, \bu_2 \}$ such that
\begin{gather}
\bS \bv = \0,\; \bA \bv \le \bb,\label{opt1}\\
\bA^T \bu_1 + \bS^T \bu_2 = \bc,\; \bu_1 \ge \0,\label{opt2}\\
\bu_1^T (\bA \bv - \bb) = 0. \label{opt3}
\end{gather}
Note that each component of $\bu_1$ can be positive or zero, and we can use this information to find a set of reactions that are forced to be inactive under optimization, as follows.
For any given optimal solution $\bv_0$, Eq.~\eqref{opt3} is equivalent to $u_{1i} (\ba_i^T \bv_0 - b_i) = 0$, $i=1,\ldots, K,$ where $u_{1i}$ is the $i$th component of $\bu_1$.
Thus, if $u_{i1}>0$ for a given $i$, we have $\ba_i^T\bv_0 = b_i$, and we say that the constraint $\ba_i^T\bv \le b_i$ is {\em binding} at $\bv_0$.
In particular, if an irreversible reaction ($v_i \ge 0$) is associated with a positive dual variable ($u_{1i} > 0$), then the irreversibility constraint is binding, and the reaction is inactive ($v_i = 0$) at $\bv_0$.
In fact, we can say much more: we prove the following theorem stating that such a reaction is actually {\em required to be inactive for all possible optimal solutions} for a given objective function $\bc^T \bv$.
\begin{theorem}\label{thm:optimal}
Suppose $\{ \bu_1, \bu_2 \}$ is a dual solution corresponding to an optimal solution of problem~\eqref{primal}.
Then, the set $\mopt$ of all optimal solutions of problem~\eqref{primal} can be written as
\begin{equation}\label{opt}
\mopt = \{ \bv \in M \,|\,
\ba_i^T \bv = b_i \text{ for all $i$ for which $u_{1i} > 0$} \},
\end{equation}
and hence every reaction associated with a positive dual component
is binding for all optimal solutions in $\mopt$.
\end{theorem}
\begin{proof}[Sketch of proof]
Let $\bv_0$ be the optimal solution associated with $\{ \bu_1, \bu_2 \}$ and let $Q$ denote the right hand side of \eqref{opt}.
Any $\bv \in Q$ is an optimal solution of problem~\eqref{primal}, since straightforward verification shows that it satisfies (\ref{opt1}-\ref{opt3}) with the same dual solution $\{ \bu_1, \bu_2 \}$.
Thus, we have $Q \subseteq \mopt$.
Conversely, suppose that $\bv$ is an optimal solution of problem~\eqref{primal}.
Then, $\bv$ can be shown to belong to $H$, which we define to be the hyperplane that is orthogonal to $\bc$ and contains $\bv_0$, i.e.,
\begin{equation}
H := \{ \bv \in \R^N \,|\, \bc^T (\bv - \bv_0) = 0 \}. \nonumber
\end{equation}
This, together with the fact that $\bv$ satisfies $\bS \bv = \0$ and $\bA \bv \le \bb$, from \eqref{opt1}, can be used to show that $\bv \in Q$.
Therefore, any optimal solution must belong to $Q$.
Putting both directions together, we have $Q = \mopt$.
\end{proof}

As an example, consider the five-reaction network shown in Fig.~\ref{fig:example}(a) where the flux $v_4$ is maximized.  The problem can be written in the form of problem~\eqref{primal} with
\begin{equation}
\bA=\begin{pmatrix*}[r]
 1 &  0 &  0 &  \,\,\,\, 0 &  0\\
-1 &  0 &  0 &  0 &  0\\
 0 & -1 &  0 &  0 &  0\\
 0 &  0 & -1 &  0 &  0\\
 0 &  0 &  0 &  0 & -1
\end{pmatrix*}\!,\,\,\,
\bb = \begin{pmatrix*}[r]
1\\ -1\\ 0\\ 0\\ 0
\end{pmatrix*}\!,\,\text{and }
\bS=\begin{pmatrix*}[r]
1 & -1 &  0 & -1 &  0\\
0 &  1 & -1 &  0 & -1
\end{pmatrix*}\!. \nonumber
\end{equation}
Note that the equality constraint $v_1=1$ is split into two inequality constraints $v_1 \ge 1$ and $v_1 \le 1$ for convenience and corresponds to the first two rows of $\bA$ and $\bb$.
The optimal solution space $\mopt$ consists of the single point $\bv = (1,0,0,1,0)^T$ and a possible choice of corresponding dual solution $\{\bu_1,\bu_2\}$ is given by
$$\bu_1=(1,0,1,0,0)^T \text{ and } \bu_2=(-1,0)^T.$$
For this dual solution, we have $u_{11}=1>0$ (corresponding to the constraint $v_1\le 1$) and $u_{13}=1>0$ (corresponding to the constraint $-v_2\le 0$), and Eq.~\eqref{opt} in Theorem~\ref{thm:optimal} becomes
\[\mopt = \{ \bv \in M \,|\, v_1=1, v_2 = 0\} = \{ \bv \in M \,|\, v_2 = 0\}.\]
Note that the constraint $v_1=1$ can be omitted since it is satisfied by any $\bv \in M$ in this example.
\begin{figure}
\begin{center}
\epsfig{figure=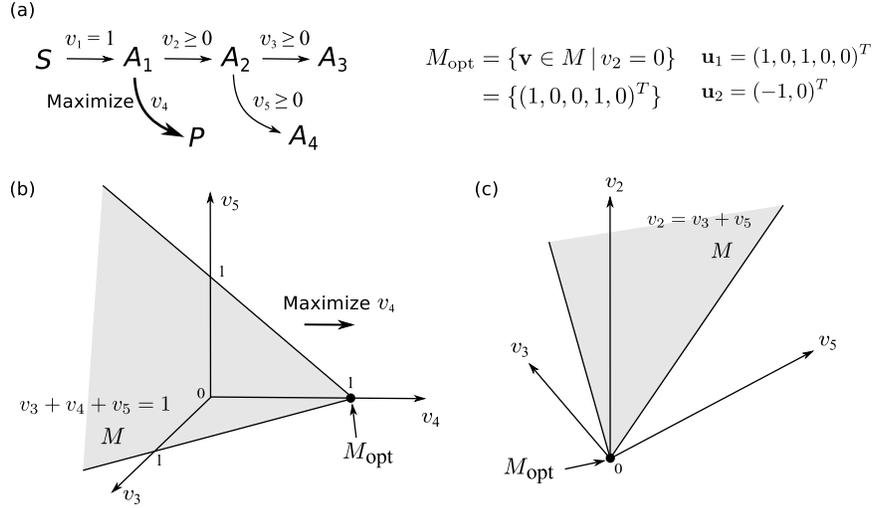,width=4.5in}
\end{center}
\caption{Simple example problem with five reactions.  (a) Reaction network, where the flux $v_4$ producing $P$ is maximized.  (b) Feasible solution space $M$ and optimal solution space $\mopt$ in the projection of the flux space onto the $(v_3,v_4,v_5)$-coordinates. (c) $M$ and $\mopt$ in the projection of the flux space onto the $(v_2,v_3,v_5)$-coordinates. }
\label{fig:example}
\end{figure}

Once we solve problem~\eqref{primal} numerically and obtain a \emph{single} pair of primal and dual solutions ($\bv_0$ and $\{ \bu_1, \bu_2 \}$), we can use the characterization of $\mopt$ given in Eq.~\eqref{opt} to identify all reactions that are required to be inactive (or active) for any optimal solutions.
To do this we solve the following auxiliary linear optimization problems for each $i=1,\ldots,N$:
\begin{equation}\label{aux}
\begin{alignedat}{2}
&\text{maximize/minimize: } & \quad & v_i \\
&\text{subject to: } & &\bS \bv = \0,\; \bA \bv \le \bb,\;
\ba_i^T\bv = b_i \text{ for all $i$ for which $u_{1i}>0$.}
\end{alignedat} \nonumber
\end{equation}
If the maximum and minimum of $v_i$ are both zero, then the corresponding reaction is required to be inactive for all $\bv \in \mopt$.
If the minimum is positive or maximum is negative, then the reaction is required to be active.
Otherwise, the reaction may be active or inactive, depending on the choice of an optimal solution.
Thus, we obtain the numbers $n_+^\text{opt}$ and $n_0^\text{opt}$ of internal and transport reactions that are required to be active and inactive, respectively, for all $\bv \in \mopt$.
The number of active reactions for any $\bv \in \mopt$ is then bounded as
\begin{equation}\label{naopt-s}
n_+^\text{opt} \le n_+(\bv) \le n - n_0^\text{opt}.
\end{equation}

The distribution of $n_+(\bv)$ within the bounds is singular: the upper bound in Eq.~\eqref{naopt-s} is attained for almost all $\bv \in \mopt$.
To see this, we apply Theorem~\ref{thm2} with $M$ replaced by $\mopt$.
This is justified since we can obtain $\mopt$ from $M$ by simply imposing additional equality constraints.
Therefore, if we set aside the $n_0^\text{opt}$ reactions that are required to be inactive
(including
$n_0^m$ and $n_0^e$ reactions that are inactive for all $\bv \in M$), all the other reactions are active for almost all $\bv \in \mopt$.
Consequently,
\begin{equation}\label{qi}
n_+(\bv) = n - n_0^\text{opt}\quad \text{for almost all } \bv \in \mopt.
\end{equation}

We can also use Theorem~\ref{thm:optimal} to further classify those inactive reactions caused by the optimization as due to two specific mechanisms:
\begin{enumerate}
\item \emph{Irreversibility.} The irreversibility constraint ($v_i \ge 0$) on a reaction can be binding ($v_i = 0$), which directly forces the reaction to be inactive for all optimal solutions.  Such inactive reactions are identified by checking the positivity of dual components ($u_{1i}$).
\item \emph{Cascading.} All other reactions that are required to be inactive for all $\bv \in \mopt$ are due to a cascade of inactivity triggered by the first mechanism, which propagates over the metabolic network via the stoichiometric constraints.
\end{enumerate}
These inactive reactions occur in addition to the irreversibility/cascading-induced reaction inactivation identified in Sec.\ \ref{sec4} for typical (in fact all) steady states.

In general, a given solution of problem~\eqref{primal} can be associated with multiple dual solutions.
The set and the number of positive components in $\bu_1$ can depend on the choice of a dual solution, and therefore the categorization according to these specific mechanisms is generally not unique.
For the example problem of Fig.~\ref{fig:example}, it is clear from $\mopt = \{ (1,0,0,1,0)^T \}$ that three reactions $v_2$, $v_3$, and $v_5$ are required to be inactive in the optimal state.
The dual solution given above categorizes $v_2$ under irreversibility, and $v_3$ and $v_5$ under cascading.  This reflects the fact that making $v_2=0$ under the stoichiometric constraint $v_2=v_3+v_5$, along with the irreversibility constraints $v_3\ge 0$ and $v_5 \ge 0$, forces $v_3=v_5=0$ [Fig.~\ref{fig:example}(c)].  Another possible dual solution is given by
$$\bu_1=(1,0,0,1,1)^T \text{ and } \bu_2=(-1,-1)^T,$$
which leads to an alternative characterization of the same $\mopt$:
$$ \mopt = \{ \bv \in M \,|\, v_3 = v_5 = 0\}, $$
categorizing $v_3$ and $v_5$ under irreversibility, and $v_2$ under cascading.
One can indeed see graphically in the projection onto $(v_3,v_4,v_5)$-coordinates in Fig.~\ref{fig:example}(b) that the maximization of $v_4$ under the stoichiometric constraint $v_3+v_4+v_5=1$ (which follows from $v_1=1$ and $\bS \bv = \0$) forces the irreversible fluxes $v_3$ and $v_5$ to be zero, which in turn forces $v_2 = 0$ [Fig.~\ref{fig:example}(c)].
Clearly, one can also characterize the optimal solution space as
$$ \mopt = \{ \bv \in M \,|\, v_2 = v_3 = v_5 = 0\}, $$
which corresponds for example to the choice of dual solution given by
$$\bu_1=(1,0,\textstyle\frac{1}{2},\frac{1}{2},\frac{1}{2})^T \text{ and } \bu_2=(-1,-\frac{1}{2})^T.$$
This leads to the categorization of all three inactive reactions under irreversibility.
Thus, we can interpret the non-uniqueness of the categorization as the fact that different sets of triggering inactive reactions can create the same cascading effect on the reaction activity.

The results above are important both because they can be applied to any linear objective function and because a significant fraction of real metabolic reactions are irreversible.  In the case of the \emph{E.\ coli}
(human) reconstructed network,  a total of 73.4\% (65.6\%) of all internal and transport reactions are irreversible. Moreover, a number of other reactions are effectively irreversible because the irreversibility of different reactions in the same pathway constrains them not to run in one of the two directions;
this leads to a total of 94.3\% (74.9\%) of  the reactions whose fluxes are either necessarily nonnegative  or necessarily nonpositive in all steady-state solutions.  In the case of growth-maximizing states for the conditions considered in our numerical experiments,  out of all 922 (3328) internal and transport reactions in the reconstructed network, a total of 146 (106) reactions are inactive due to irreversibility constraints, and a total of 114 (293) other reactions are inactive  due to a cascade of reaction inactivation; some of the irreversible reactions can be assigned to either the first or the second of these two groups. The bounds provided by Eq.\ (\ref{naopt-s}) depend on the objective function. In the case of growth-maximizing states, the lower and upper bounds are 273 (113) and 339 (1180), respectively. These numbers should be compared with the number 599 (1579) of reactions that are active in typical, suboptimal states.  This clearly shows that optimal states are necessarily constrained to have a smaller number of active reactions, and that this is due to the presence of irreversible reactions in the network.

%--------------------------------------------------------------------------------------------------
\section{Typical linear objective functions}
\label{sec6}

Another problem of interest concerns the uniqueness of the optimal solutions.
While a number of necessary and/or sufficient conditions for this uniqueness are known~\cite{MR2378114,Mangasarian:1979fk}, we are not aware of any probabilistic statements in the literature addressing this issue.
Since the feasible solution space $M$ is convex, its ``corners'' can be mathematically formulated as  {\em extreme points}, defined as points $\bv \in M$ that cannot be written as $\bv = a \bx + b \by$ with $a+b=1$, $0<a<1$ and ${\bf x, y} \in M$ such that ${\bf x \neq y}$.
Intuition from the two-dimensional case (Fig.~\ref{fig:extreme}) suggests that for a typical choice of the objective vector $\bc$ such that problem~\eqref{primal} has a solution, the solution is unique and located at an extreme point of $M$.
\begin{figure}
\begin{center}
\epsfig{figure=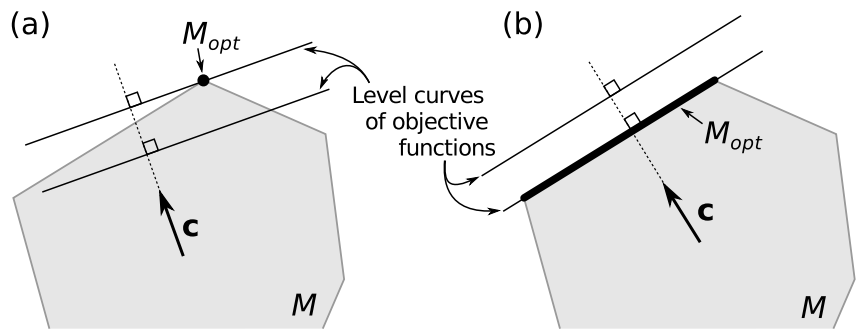,width=4.5in}
\end{center}
\caption{Uniqueness of optimal solutions.  (a) For a typical objective vector $\bc$, a
unique optimum is obtained at a single extreme point.  (b) In the exceptional cases where $\bc$ is perpendicular to an edge, all points on the edge are optimal.
}
\label{fig:extreme}
\end{figure}
We prove here that this is indeed true in general, as long as the objective function is bounded on $M$, and hence an optimal solution exists.
\begin{theorem}\label{thm:corner}
Suppose that the set of objective vectors
$$B = \{\bc \in \R^N \,|\, \text{$\bc^T \bv$ is bounded on $M$}\}$$
 has positive Lebesgue measure.
Then, for almost all $\bc$ in $B$, there is a unique solution of problem~\eqref{primal}, and it is located at an extreme point of $M$.
\end{theorem}
\begin{proof}
For a given $\bc \in B$, the function $\bc^T \bv$ is bounded on $M$, so the solution set $\mopt = \mopt(\bc)$ of problem~\eqref{primal} consists of either a single point or multiple points.
Suppose $\mopt$ consists of a single point $\bv$ and it is not an extreme point.
By definition, it can be written as $\bv = a{\bf x}+b{\bf y}$ with $a+b=1$, $0<a<1$ and ${\bf x, y} \in M$ such that ${\bf x \neq y}$.
Since $\bv$ is the only solution of problem~\eqref{primal}, ${\bf x}$ and ${\bf y}$ must be suboptimal, and hence we have $\bc^T {\bf x} < \bc^T \bv$ and $\bc^T {\bf y} < \bc^T \bv$.
Then,
\begin{eqnarray*}
\bc^T {\bf y} &=& (\bc^T \bv - a \bc^T {\bf x})/b\\
&>& (\bc^T \bv - a \bc^T \bv)/b\\
&=& \bc^T \bv,
\end{eqnarray*}
and we have a contradiction with the fact that $\bv$ is optimal.
Therefore, if $\mopt$ consists of a single point, it must be an extreme point of $M$.

We are left to show that the set of $\bc \in B$ for which $\mopt(\bc)$ consists of multiple points has Lebesgue measure zero.
By Theorem~\ref{thm:optimal}, for a given $\bc$, there exists a set of indices $I \subseteq \{1,\ldots,K\}$ such that $\mopt(\bc) = Q_I := \{ \bv \in M \,|\, \ba_i^T \bv = b_i \text{ for all $i \in I$} \}$, so
\begin{equation}\label{mopt}
\{ \bc \in \R^N \,|\, \mopt(\bc)\text{ contains multiple points} \}
\subseteq \bigcup_{I} \{\bc \in \R^N \,|\, Q_I = \mopt(\bc) \},
\end{equation}
where the union is taken over all $I \subseteq \{1,\ldots,K\}$ for which $Q_I$ contains multiple points.
If $\bc$ is in one of the sets in the union in Eq.~\eqref{mopt}, the set $Q_I$, being the set of all optimal solutions, is orthogonal to $\bc$.  Hence, $\bc$ is in $Q_I^\perp$, the orthogonal complement of $Q_I$ defined as the set of all vectors orthogonal to $Q_I$.
Therefore,
\begin{equation}\label{mopt2}
\{ \bc \in \R^N \,|\, \mopt(\bc)\text{ contains multiple points} \}
\subseteq \bigcup_{I} Q_I^\perp,
\end{equation}
Because $Q_I$ is convex, it contains multiple points if and only if its dimension is at least one, implying that each $Q_I^\perp$ in the union in Eq.~\eqref{mopt2} has dimension at most $N-1$, and hence has zero Lebesgue measure in $\R^N$.
Since there are only a finite number of possible choices for $I \subseteq \{1,\ldots,K\}$, the right hand side of Eq.~\eqref{mopt2} is a finite union of sets of Lebesgue measure zero.
Therefore, the left hand side also has Lebesgue measure zero.
\end{proof}

Note that growth rate is not a typical objective function. Because this objective function has nonrandom coefficients and involves only a fraction of all metabolic fluxes,
the objective vector {\bf c} is generally perpendicular to a surface limiting the space of feasible solutions.
For this reason, the growth-maximizing states are generally not unique. In the case of the \emph{E.\ coli} reconstructed network simulated in glucose minimal medium, our numerical calculations indicate that the growth-maximizing solutions form a space that is $26$-dimensional. In the case of the human reconstructed network, the corresponding dimension is 494.

\section{Numerical experiments}
\label{sec7}

Two questions follow from the results above. First, given the \emph{a priori} surprising finding that metabolic activity as measured by the number of active reactions decreases in optimal states, what happens if we use other measures of metabolic activity such as total reaction flux in the network? Second, given that these results were derived for linear objective functions, to what extent does the observed reduction in the number of active reactions manifest itself in nonlinear objective functions of biological significance? These two questions are best examined using numerical experiments.  Both are addressed below by considering the following objective functions:

\begin{enumerate}
\item \emph{Total flux in the network}: Defined as $\phi_{f}= \sum_i |v_i|$,  where $i$ runs through all internal and transport reactions (excluding the biomass flux), it measures the overall metabolic activity
while accounting for the differences in the fluxes of different reactions. This objective function is nonlinear because of the absolute value used
to properly measure the flux of the reversible reactions.

\item \emph{Total mass flow in the network}:  Defined as $\phi_{m}= \sum_i m_i |v_i|$, where $m_i$  is the mass of the reactants (or products) involved in reaction $i$,  it measures the overall metabolic activity weighted by the mass flow of each reaction \cite{kritz2010}. The sum is over the same reaction set considered in the definition of  $\phi_{f}$.
\end{enumerate}
For other nonlinear objective functions of biological significance in cellular metabolism, we refer to Ref.\ \cite{sch2007}.

\begin{figure}
\begin{center}
\epsfig{figure=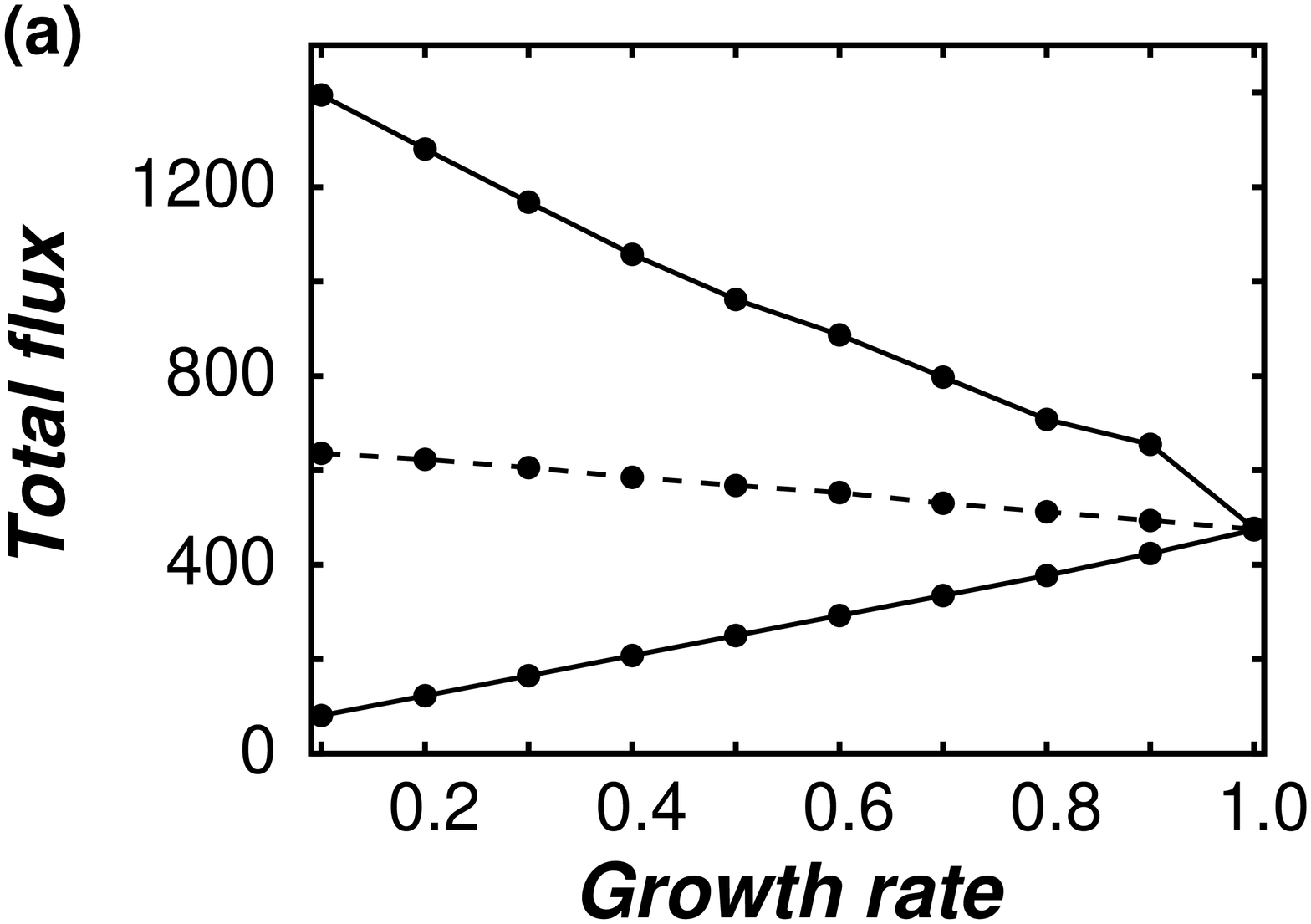,angle=-90, width=6.0cm}
\epsfig{figure=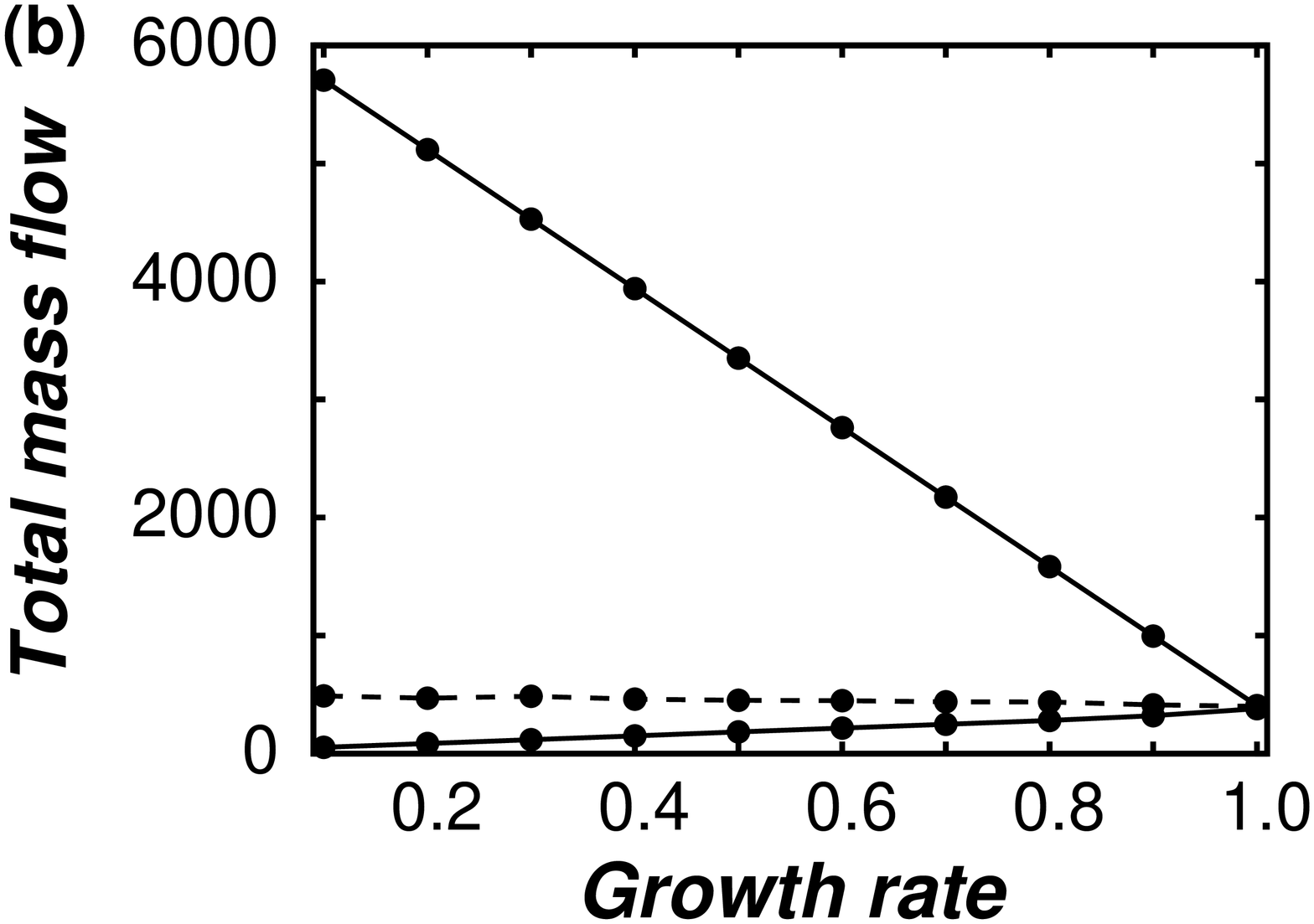,angle=-90, width=6.0cm}
\epsfig{figure=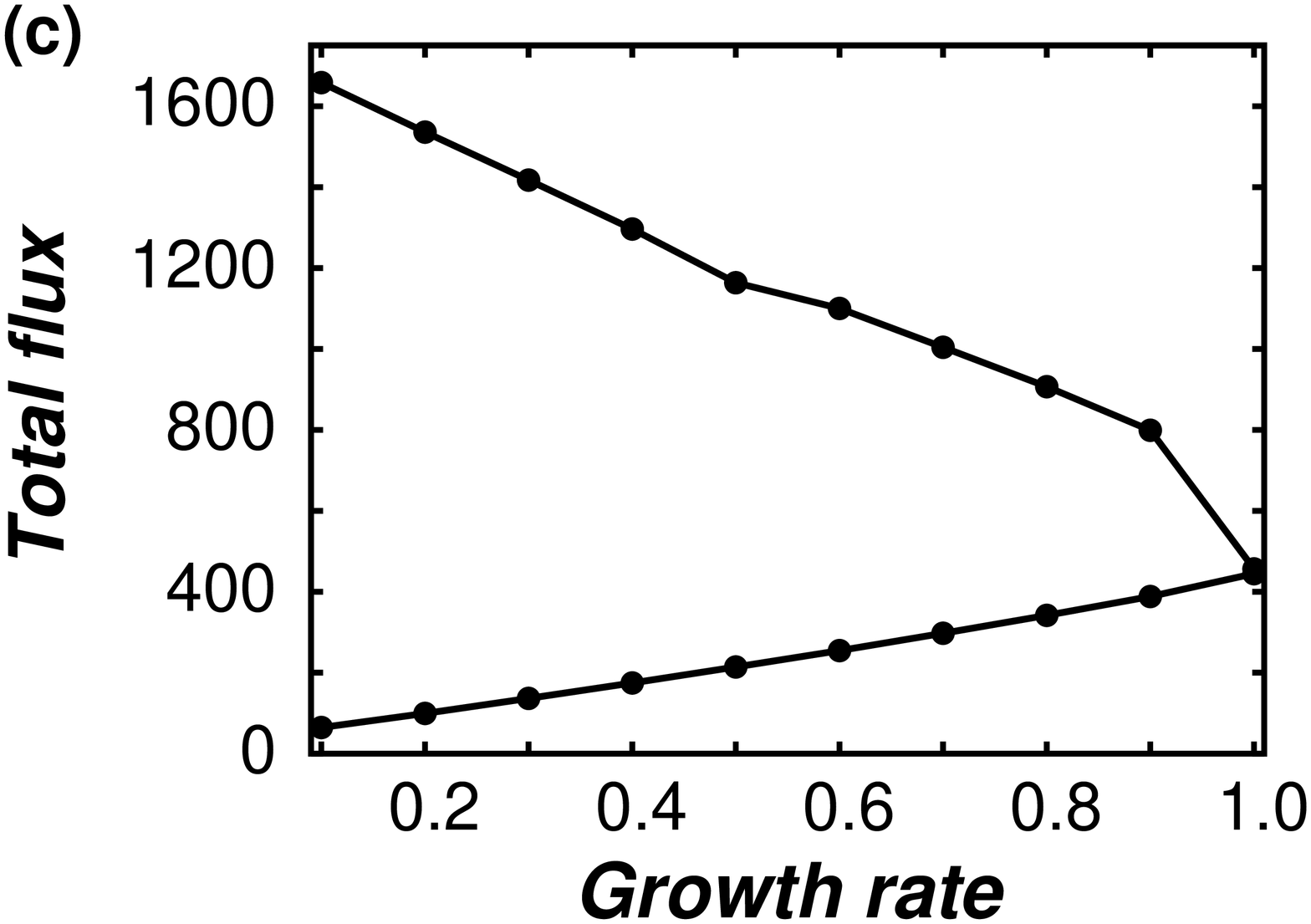,angle=-90, width=6.0cm}
\epsfig{figure=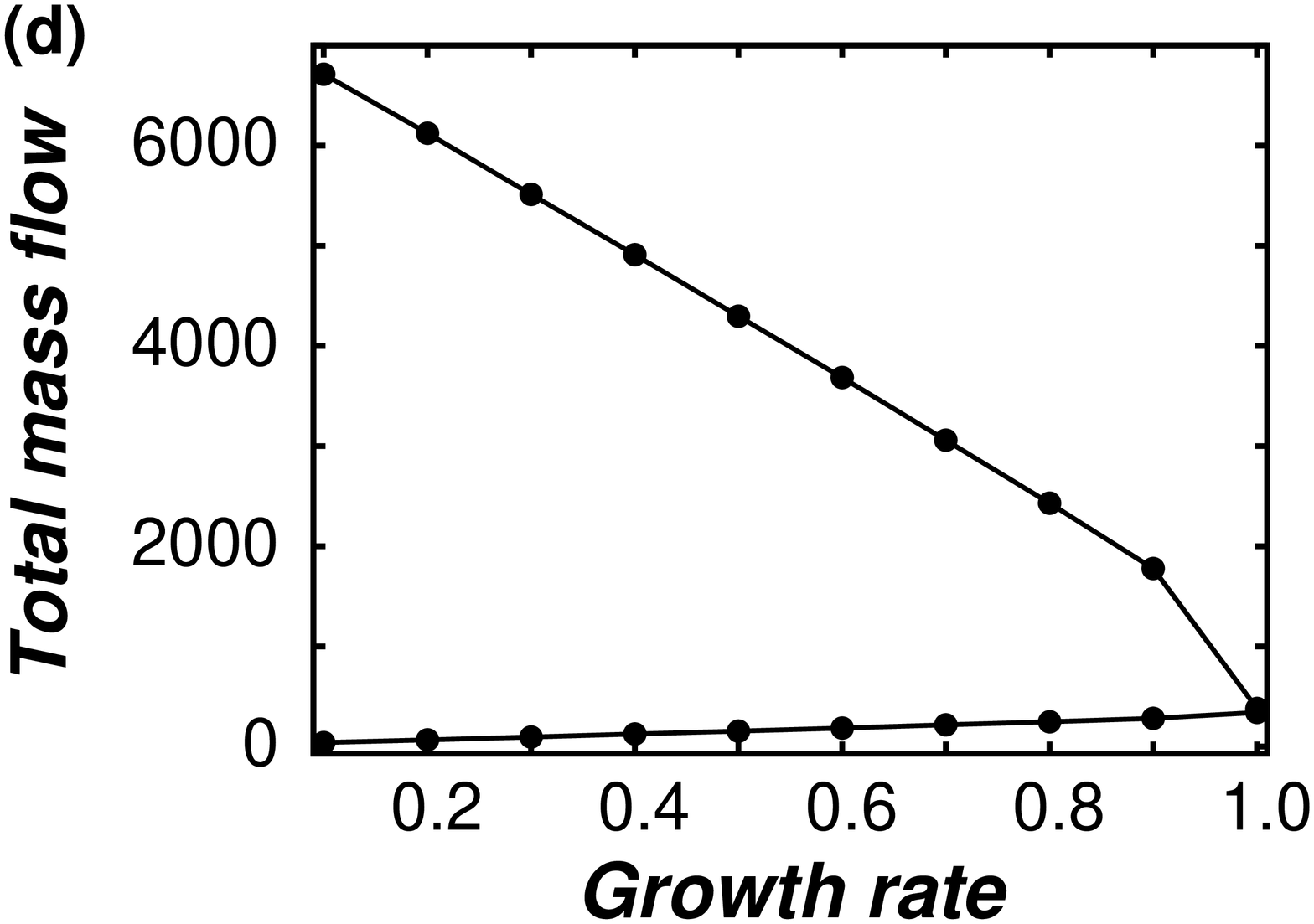,angle=-90, width=6.0cm}
\end{center}
\caption{Metabolic activity measured in terms of total flux and total mass flow for  the \emph{E.\ coli}  reconstructed network simulated in glucose minimal medium. (a,b) Maximum and minimum of the total flux $\phi_{f}$ in units of mmol/g DW-h  (a) and of the total mass flow $\phi_{m}$ in units of g/g DW-h (b) (continuous lines) as a function of the growth rate (normalized by its maximum under the given conditions), where g DW denotes grams of dry weight. The dashed lines indicate the average in the space of feasible solutions calculated from  $5\times 10^6$ randomly selected points obtained using the hit-and-run method. The standard deviation
in estimating the average
 is smaller than the size of the symbols.
(c,d) Maximum and minimum $\phi_{f}$  (c) and  $\phi_{m}$  (d) when the irreversibility constraints $\alpha_i$ are relaxed from zero to $-0.1$ mmol/g DW-h for all irreversible reactions. Allowing all the reactions to be reversible leads to moderate changes in the optimal values of the objective functions, but it leads to drastic changes in the number of active reactions (Table \ref{tabl}).
}
\label{fig:numerical}
\end{figure}

Figure \ref{fig:numerical} shows the results of our numerical experiments for $\phi_{f}$ and $\phi_{m}$ on the \emph{E.\ coli}  reconstructed network. In both cases, the maximum (minimum) of the objective function for a given growth rate decreases (increases) as the growth rate increases, and it converges to essentially a single intermediate value for all states that maximize growth rate [Fig.\ \ref{fig:numerical}(a,b)].  The average activity, which we determined by randomly sampling the solution space, is essentially constant (it decreases very slowly as the growth rate increases). The sampling of the solution space was performed using the hit-and-run method~\cite{smith84}, which is an efficient algorithm to sample  high-dimensional convex regions. Our implementation of this method is as described in our previous study \cite{Cornelius2011} and involves  artificial centering \cite{kauf98}.  The observed behavior of the total flux activity and total mass flow activity should be contrasted with metabolic activity as measured in terms of the number of active reactions, which is significantly smaller at growth-maximizing states.

The average number of active reactions in the $\phi_{f}$-maximizing states across different growth rates is just 302 out of a total of 571 that would be active in typical states (the latter too is an average over different growth rates, and  is smaller than the number 599 anticipated in Sec.\ \ref{sec5} because of  additional constraints set to the diverging cycles throughout this section---see below).  A similar result holds true for $\phi_{m}$-maximizing states (Table \ref{tabl}).  Therefore, the maximizations of total flux and total mass flow in the network also lead to a reduced (rather than increased) number of active reactions compared to typical states [such as those determined by the hit-and-run method in Fig.\ \ref{fig:numerical}(a,b)]. This number varies very little with growth rate and is essentially undistinguishable from the number of active reactions in growth-maximizing states (see standard deviations in Table \ref{tabl}).  While these results concern {\it E.\ coli}, we note that similar trends are observed for the human metabolic network.

If the irreversible reactions are made reversible [Fig.\ \ref{fig:numerical}(c,d)],  then the number of active reactions increases.
For states that minimize $\phi_{f}$  and $\phi_{m}$, the number of active reactions  jumps to a large number when $\alpha_i$ of the  irreversible reactions is assigned to be just slightly negative, and then decreases as the irreversibility constraints are further relaxed (Table \ref{tabl}).
For states that maximize $\phi_{f}$  and $\phi_{m}$, the increase in the number of active reactions is by a factor of nearly $2$ (Table \ref{tabl}).   This number  is comparable to the number of active reactions in typical suboptimal states of the original network.
Therefore, like in the case of linear objective functions, the reduced number of active reactions found in states that maximize the total flux or the total mass flow is due to the presence of irreversible reactions in the metabolic network.

{\small
\begin{table}[!ht]
\caption{Number of active reactions in states maximizing or minimizing the total flux, $\phi_{f}$, and the total mass flow, $\phi_{m}$.
The relaxation of the irreversibility constraints is implemented by allowing  $\alpha_i$ to be negative for all irreversible reactions, as indicated in the leftmost column.
Each column shows the average and standard deviation calculated over the growth rates considered in Fig.\ \ref{fig:numerical}.}
\begin{tabular}{l|c|c|c|c}
\hline
  & $\max \phi_f$  & $\min \phi_f$  & $ \max \phi_m$ & $\min \phi_m$ \\
\hline
\hline
{\small actual irreversibility} &   301.9 (2.6) & 292.3 (2.7) & 300.4 (2.7) & 292.1 (2.6)\\
\hline
irreversibility relaxed to $-10^{-3}$ & 580.5 (8.6)   & 506.8 (14.5) &   593.2 (8.1) &  501.1 (17.4)\\
irreversibility relaxed to $-10^{-2}$ & 592.7 (4.5)   & 502.2 (17.8)  &   597.9 (2.5) &  502.1 (17.7)  \\
irreversibility relaxed to $-10^{-1}$ & 595.4 (4.2)   & 469.4 (38.0) &   596.6 (2.2) &  475.8 (32.0) \\
irreversibility relaxed to $-1$           & 595.0 (4.0) & 370.1 (35.7) & 596.7 (3.5) & 404.0 (41.2)  \\
\hline
\end{tabular}
\begin{flushleft}
\end{flushleft}
\label{tabl}
\end{table}
}

All simulations  presented in this paper are based on a reconstructed metabolic network of  \emph{E.\ coli} K-12, which represents a further curated version of the iJR904 model  \cite{Reed2003-s} in which duplicated reactions have been removed, and on the most complete reconstructed human metabolic network, generated by applying the same curation to the \emph{Homo sapiens} Recon 1 model \cite{duarde2007}. The \emph{E.\ coli} (human) network used  consists of 922 (3328) reactions, 901 (1491) enzyme- and transport protein-coding genes, 618 (2766) metabolites, 143 (404) exchange fluxes, and the biomass flux. For the \emph{E.\ coli} network, the simulated medium  had limited amount of glucose (10 mmol/g DW-h) and oxygen (20 mmol/g DW-h), and unlimited amount of  sodium, potassium,  carbon dioxide, iron (II), protons,  water, ammonia, phosphate, and sulfate; the flux through the ATP maintenance reaction was set to 7.6 mmol/g DW-h. For the human network, we used a medium with limited amount of glucose (1 mmol/g DW-h) and unlimited amount of  oxygen, sodium, potassium, calcium, iron (II and III),  protons, water, ammonia, chlorine, phosphate, and sulfate; for the biomass composition, we followed Ref.\ \cite{shlomi2011}. In our simulations, $10^{-6}$ mmol/g DW-h  was used as a  flux threshold to define the set of reactions considered active. A few cycles whose flux or mass flow would diverge in the optimization of the corresponding objective function were assigned the  minimum feasible flux in the optimization of $\phi_{f}$ and the minimum feasible mass flow in the optimization of $\phi_{m}$ (under the constraint of not altering the fluxes of the other reactions).  These minimum flux values were also adopted as bounds in our hit-and-run sampling. All numerical calculations were implemented using the COBRA Toolbox \cite{cobra} and the CPLEX optimization software \cite{cplex}.

\section*{Acknowledgments}
This study was supported by the National Science Foundation under Grant DMS-1057128, the National
Cancer Institute under Grant 1U54CA143869-01, and a Sloan Research Fellowship to A.E.M.

\medskip
\medskip


\begin{thebibliography}{99}

\bibitem{alon2006}
U. Alon, ``An Introduction to Systems Biology: Design Principles of Biological Circuits," Chapman and Hall/CRC, Boca Raton, FL, 2006. 

\bibitem{bar2004}
A.-L. Barab\'asi and Z.N. Oltvai,  
{\it Network biology: understanding the cell's functional organization},
Nat. Rev. Genet. {\bf 5} (2004), 101--113.

\bibitem{cobra}
S.D. Becker, {\it et al.}, 
{\it  Quantitative prediction of cellular metabolism with constraint-based models: The COBRA Toolbox}, 
Nat. Protoc. {\bf 2} (2007), 727--738.

\bibitem{math_ref}
M.J. Best and K. Ritter,
``Linear Programming: Active Set Analysis and Computer Programs,"
Prentice-Hall, Engelwood Cliffs, NJ, 1985.

\bibitem{blank2005}
L.M. Blank, L. Kuepfer and U. Sauer, 
 {\it Large-scale 13C-flux analysis reveals mechanistic principles of metabolic network robustness to null mutations in yeast}, 
 Genome Biol. {\bf 6} (2005), R49.
 
\bibitem{bonarius1997}
H.P.J. Bonarius, G. Schmid and J. Tramper,
{\it Flux analysis of underdetermined metabolic networks: The quest for the missing constraints},
Trends Biotechnol. {\bf 15} (1997), 308--314.

\bibitem{Cornelius2011}
S.P. Cornelius, J.S. Lee and A.E. Motter,
{\it Dispensability of Escherichia coli's latent pathways},
Proc. Natl. Acad. Sci. USA {\bf 108} (2011), 3124--3129. 

\bibitem{duarde2007}
N.C. Duarte, {\it et al.},
{\it Global reconstruction of the human metabolic network based on genomic and bibliomic data}, 
Proc. Natl. Acad. Sci. USA {\bf 104} (2007), 1777--1782.

\bibitem{Fong2005}
S.S. Fong, A.R. Joyce and B.{\O}.  Palsson,
{\it Parallel adaptive evolution cultures of {\it Escherichia coli} lead to convergent growth phenotypes with different gene expression states}, 
Genome. Res. {\bf 15}  (2005), 1365--1372.

\bibitem{Fong2006}
S.S. Fong, A. Nanchen, B.{\O}. Palsson and U. Sauer,
{\it Latent pathway activation and increased pathway capacity enable {\it Escherichia coli} adaptation to loss of key metabolic enzymes},
J. Biol. Chem. {\bf 281} (2006), 8024--8033.

\bibitem{cplex}
ILOG CPLEX (Version 10.2.0). Available: \texttt{http://www.ilog.com/products/cplex/}.

\bibitem{kauf98}
D.E. Kaufman and R.L. Smith,  
{\it Direction choice for accelerated convergence in hit-and-run sampling}, 
Oper. Res. {\bf 46} (1998), 84--95.

\bibitem{Kim2009}
D.-H. Kim and A.E. Motter,
{\it Slave nodes and the controllability of metabolic networks}, 
New J. Phys. {\bf 11} (2009), 113047.

\bibitem{kritz2010}
M.V. Kritz, M.T. dos Santos, S. Urrutia and J.-M. Schwartz, 
{\it Organising metabolic networks: Cycles in flux distributions}, 
J. Theo.  Biol. {\bf 265} (2010), 250--260.

%\bibitem{glpk}
%A. Makhorin,
%``GNU Linear Programming Kit (GLPK)'', 2001.
%Available: \url{http://www.gnu.org/software/glpk/glpk.html} 

\bibitem{Mangasarian:1979fk}
O.L. Mangasarian,
{\it Uniqueness of solution in linear programming},
Linear Algebra Appl. {\bf 25} (1979), 151--162.

\bibitem{Motter2010}
A.E. Motter,
{\it Improved network performance via antagonism: From synthetic rescues to multi-drug combinations}, 
BioEssays {\bf 32} (2010), 236--245.

\bibitem{Motter2008}
A.E. Motter, N. Gulbahce, E. Almaas and A.-L. Barab\'asi,
{\it Predicting synthetic rescues in metabolic networks},
Mol. Syst. Biol. {\bf 4} (2008), 168.

\bibitem{Takashi2008}
T. Nishikawa, N. Gulbahce and A.E. Motter,
{\it Spontaneous reaction silencing in metabolic optimization}, 
PLoS Comput. Biol. {\bf 4} (2008), e1000236. 

\bibitem{papp2004}
B. Papp, C. P\'al and L.D. Hurst, 
{\it Metabolic network analysis of the causes and evolution of enzyme dispensability in yeast}, 
Nature {\bf 429} (2004), 661--664.

\bibitem{pals2006}
B.{\O}.  Palsson,
``Systems Biology: Properties of Reconstructed Networks,"  Cambridge University Press, Cambridge, UK, 2006.

\bibitem{ravasz2002}
E. Ravasz, A. Somera, D. Mongru, Z. Oltvai and A.-L. Barab\'asi,
{\it Hierarchical organization of modularity in metabolic networks}, 
Science {\bf 297} (2002), 1551--1555.

\bibitem{Reed2003-s}
J.L. Reed,  T.D. Vo, C.H. Schilling and  B.{\O}. Palsson,
{\it An expanded genome-scale model of {\it Escherichia coli} K-12 (iJR904 GSM/GPR)},
Genome Biol. {\bf 4}  (2003),  R54.

\bibitem{rudin1987real}
W.~Rudin,
``Real and Complex Analysis,''
McGraw-Hill, Boston, MA, 1987.

\bibitem{sch2007}
R. Schuetz, L. Kuepfer and U. Sauer, 
{\it Systematic evaluation of objective functions for predicting intracellular fluxes in Escherichia coli},  
Mol. Syst. Biol. {\bf 3} (2007), 119. 

\bibitem{shlomi2011}
T. Shlomi, T. Benyamini, E. Gottlieb, R. Sharan and E. Ruppin, 
{\it Genome-scale metabolic modeling elucidates the role of proliferative adaptation in causing the Warburg effect},  
PLoS Comput. Biol. {\bf 7} (2011),  e1002018.

\bibitem{smith84}
R.L. Smith, 
{\it Efficient Monte Carlo procedures for generating points uniformly distributed over bounded regions}, 
Oper. Res. {\bf 32} (1984), 1296--1308.

\bibitem{spirin2003}
V. Spirin and L.A. Mirny, 
{\it Protein complexes and functional modules in molecular networks}, 
Proc. Natl. Acad. Sci. USA {\bf 100} (2003), 12123--12128.

\bibitem{MR2378114}
P. Szil{\'a}gyi,
{\it On the uniqueness of the optimal solution in linear programming},
Rev. Anal. Num\'er. Th\'eor. Approx. {\bf 35} (2006), 225--244.

\bibitem{varma1994}
A. Varma and B.{\O}.  Palsson,
{\it Metabolic flux balancing: Basic concepts, scientific and practical use},
Nat. Biotechnol. {\bf 12} (1994), 994--998. 


\end{thebibliography}
\end{document}